\documentclass[11pt]{article}
\usepackage[margin=1in]{geometry}
\usepackage{graphicx}
\usepackage[symbol]{footmisc}
\usepackage[colorlinks=true, allcolors=black]{hyperref}
\usepackage{amsthm,amsmath,amsfonts}
\usepackage{thm-restate}
\usepackage{soul}

\newtheorem{theorem}{Theorem}
\newtheorem{problem}[theorem]{Problem}
\newtheorem{obstacle}[theorem]{Obstacle}
\newtheorem{definition}[theorem]{Definition}
\newtheorem{lemma}[theorem]{Lemma}
\newtheorem{corollary}[theorem]{Corollary}
\newtheorem{fact}[theorem]{Fact}
\DeclareMathOperator{\poly}{poly}
\DeclareMathOperator{\total}{total}

\title{Bicriteria approximation for minimum dilation graph augmentation}
\author{Kevin Buchin, Maike Buchin, Joachim Gudmundsson, Sampson Wong}
\date{}

\begin{document}

\maketitle

\begin{abstract}
Spanner constructions focus on the initial design of the network. However, networks tend to improve over time. In this paper, we focus on the improvement step. Given a graph and a budget~$k$, which~$k$ edges do we add to the graph to minimise its dilation? Gudmundsson and Wong~[TALG'22] provided the first positive result for this problem, but their approximation factor is linear in~$k$. 

Our main result is a $(2 \sqrt[r]{2} \ k^{1/r},2r)$-bicriteria approximation that runs in $O(n^3 \log n)$ time, for all $r \geq 1$. In other words, if $t^*$ is the minimum dilation after adding any $k$ edges to a graph, then our algorithm adds~$O(k^{1+1/r})$ edges to the graph to obtain a dilation of~$2rt^*$. Moreover, our analysis of the algorithm is tight under the Erd\H{o}s girth conjecture. 
\end{abstract}

\section{Introduction}
\label{section:introduction}

Let~$G$ be a graph embedded in a metric space~$M$. Let $V(G)$, $E(G)$ be the vertices and edges of~$G$. For vertices $u,v \in V(G)$, define $d_M(u,v)$ to be the metric distance between points $u,v \in M$, and define~$d_G(u,v)$ to be the shortest path distance between vertices $u,v \in G$. The \emph{dilation} or \emph{stretch} of~$G$ is the minimum~$t \in \mathbb R$ so that for all $u,v \in V(G)$, we have $d_G(u,v) \leq t \cdot d_M(u,v)$. 

Dilation measures the quality of a network in applications such as transportation, energy, and communication. For now, we restrict our attention to the special case of low dilation trees.

\begin{problem}
\label{problem:tree_spanner}
Given a set of~$n$ points~$V$ embedded in a metric space~$M$, compute a spanning tree of~$V$ with minimum dilation.
\end{problem}

Problem~\ref{problem:tree_spanner} is known across the theory community, as either the minimum dilation spanning tree problem~\cite{DBLP:journals/comgeo/AronovBCGHSV08,DBLP:conf/cocoon/BrandtGRS15,DBLP:journals/comgeo/CheongHL08}, the tree spanner problem~\cite{DBLP:journals/siamdm/CaiC95,DBLP:journals/dam/FeketeK01,DBLP:journals/ipl/FominGL11} or the minimum maximum-stretch spanning tree problem~\cite{DBLP:journals/siamcomp/EmekP08,DBLP:journals/amc/LinL20,DBLP:conf/mfcs/Peleg02}. The problem is \mbox{NP-hard} even if $M$ is an unweighted graph metric~\cite{DBLP:journals/siamdm/CaiC95} or the Euclidean plane~\cite{DBLP:journals/comgeo/CheongHL08}. Problem~\ref{problem:tree_spanner} is closely related to tree embeddings of general metrics~\cite{DBLP:conf/soda/BadoiuIS07}, and has applications to communication networks and distributed systems~\cite{DBLP:conf/mfcs/Peleg02}. 

The approximability of Problem~\ref{problem:tree_spanner} is an open problem stated in several surveys and papers~\cite{DBLP:journals/comgeo/CheongHL08,DBLP:books/el/00/Eppstein00,DBLP:conf/mfcs/Peleg02}, and is a major obstacle towards constructing low dilation graphs with few edges~\cite{DBLP:journals/comgeo/AronovBCGHSV08,DBLP:journals/talg/GudmundssonW22}. The minimum spanning tree is an $O(n)$-approximation~\cite{DBLP:books/el/00/Eppstein00} for Problem~\ref{problem:tree_spanner}, but no better result is known. Only in the special case where $M$ is an unweighted graph is there an $O(\log n)$-approximation~\cite{DBLP:journals/siamcomp/EmekP08}.

\begin{obstacle}
\label{obstacle:tree_spanner}
Is there an $O(n^{1-\varepsilon})$-approximation algorithm for Problem~\ref{problem:tree_spanner}, for any $\varepsilon > 0$?
\end{obstacle}

If we no longer restrict ourselves to trees, we can shift our attention to spanners, which are low dilation sparse graphs. An advantage of spanners over minimum dilation trees is that spanners are not affected by Obstacle~\ref{obstacle:tree_spanner}. Spanners obtain significantly better dilation guarantees, at the cost of adding slightly more edges. The trade-off between sparsity and dilation in spanners has been studied extensively~\cite{DBLP:journals/dcg/AlthoferDDJS93,DBLP:conf/compgeom/DasHN93,DBLP:journals/siamcomp/FiltserS20,DBLP:conf/focs/LeS19}. For an overview of the rich history and multitude of applications of spanners, see the survey on graph spanners~\cite{DBLP:journals/csr/AhmedBSHJKS20} and the textbook on geometric spanners~\cite{DBLP:books/daglib/NarasimhanSmid}. 

Spanner constructions focus on the initial design of the network. However, networks tend to improve over time. In this paper, we focus on the improvement step. Given a graph and a budget~$k$, which~$k$ edges do we add to the graph to minimise its dilation? 

\begin{problem}
\label{problem:graph_augmentation}
Given a positive integer $k$ and a metric graph $G$, compute a set~$S$ of $k$~edges so that the dilation of the graph $G' = (V(G), E(G) \cup S)$ is minimised. Note that $S \subseteq V(G) \times V(G)$. 
\end{problem}

Narasimham and Smid~\cite{DBLP:books/daglib/NarasimhanSmid} stated Problem~\ref{problem:graph_augmentation} as one of twelve open problems in the final chapter of their reference textbook. For over a decade, the only positive results for Problem~\ref{problem:graph_augmentation} were for the special case where $k=1$~\cite{DBLP:journals/josis/AronovBBJJKLLSS11,DBLP:journals/siamcomp/FarshiGG08,DBLP:conf/isaac/LuoW08,DBLP:journals/comgeo/Wulff-Nilsen10}. In 2021, Gudmundsson and Wong~\cite{DBLP:journals/talg/GudmundssonW22} showed the first positive result for $k \geq 2$, by providing an~\mbox{$O(k)$-approximation} algorithm that runs in $O(n^3 \log n)$ time. A downside of~\cite{DBLP:journals/talg/GudmundssonW22} is that their approximation factor is linear in~$k$. However, since Problem~\ref{problem:tree_spanner} is a special case of Problem~\ref{problem:graph_augmentation}, Obstacle~\ref{obstacle:tree_spanner} applies to Problem~\ref{problem:graph_augmentation} as well.

\begin{obstacle}
\label{obstacle:graph_augmentation}
One cannot obtain an~\mbox{$O(k^{1-\varepsilon})$-approximation} algorithm for Problem~\ref{problem:graph_augmentation} for any $\varepsilon > 0$, without first resolving Obstacle~\ref{obstacle:tree_spanner}. 
\end{obstacle}

One way to circumvent Obstacle~\ref{obstacle:graph_augmentation} is to consider a bicriteria approximation. An advantage of a bicriteria approximation is that we can obtain significantly better dilation guarantees, at the cost of adding slightly more edges. 

The goal of our bicriteria problem is to investigate the trade-off between sparsity and dilation. We define the sparsity parameter~$f$ to be the number of edges added by our algorithm divided by~$k$. We define the dilation parameter~$g$ to be the dilation of our algorithm (which adds~$fk$ edges) divided by the dilation of the optimal solution (which adds~$k$ edges). 

\begin{problem}
\label{problem:bicriteria}
Given a positive integer~$k$, a metric graph~$G$, sparsity $f \in \mathbb R$ and dilation $g \in \mathbb R$, construct a set~$S$ of~$fk$ edges so that the dilation of the graph $G' = (V(G), E(G) \cup S)$ is at most~$gt^*$, where~$t^*$ is the minimum dilation in Problem~\ref{problem:graph_augmentation}. Note that~$S \subseteq V(G) \times V(G)$. 
\end{problem}

We define an~$(f,g)$-bicriteria approximation to be an algorithm for Problem~\ref{problem:bicriteria} that achieves sparsity~$f$ and dilation~$g$.

\subsection{Contributions}
\label{subsection:contributions}

Our main result is a \mbox{$(2 \sqrt[r]{2} \ k^{1/r},2(1+\delta)r)$}-bicriteria approximation for Problem~\ref{problem:bicriteria} that runs in $O(n^3 (\log n + \log \frac 1 \delta))$ time, for all~$r \geq 1$ and $\delta > 0$. In other words, if $t^*$ is the minimum dilation after adding any $k$ edges to a graph, then our algorithm adds~$O(k^{1+1/r})$ edges to the graph to obtain a dilation of~$2(1+\delta)rt^*$. Our dilation guarantees are significantly better than the previous best result~\cite{DBLP:journals/talg/GudmundssonW22}, at the cost of adding slightly more edges. For example, if $r = \log(2k)$ we obtain a~$(4,2(1+\delta) \log (2k))$-bicriteria approximation algorithm, which adds~$4k$ edges to the graph to obtain a dilation of $2 (1+\delta) \log(2k) t^*$. See Table~\ref{table:contributions}.

Our approach uses the greedy spanner construction. The greedy spanner is among the most extensively studied spanner constructions~\cite{DBLP:journals/dcg/AlthoferDDJS93,DBLP:conf/compgeom/DasHN93,DBLP:journals/siamcomp/FiltserS20,DBLP:conf/focs/LeS19}. Therefore, it is perhaps unsurprising that greedy spanner can be used for Problem~\ref{problem:bicriteria}. Nonetheless, we believe that our result shows the utility and versatility of the greedy spanner.

Our main technical contribution is our analysis of the greedy spanner. Our main insight is to construct an auxilliary graph, which we call the girth graph, and to argue that the approximation ratio is bounded by the length of the shortest cycle in the girth graph. Moreover, our analysis of the greedy spanner is tight, up to constant factors. In particular, assuming the Erd\H{o}s girth conjecture, there is a graph class for which our algorithm is an \mbox{$(\Omega(k^{1/r}), 2r+1)$-bicriteria} approximation.

Assuming W[1] $\neq$ FPT, we prove that one cannot obtain a~\mbox{$(h(k), 2 - \varepsilon)$-bicriteria} approximation, for any computable function~$h$ and for any $\varepsilon > 0$. Since one cannot approximate the dilation to within a factor of $(2-\varepsilon)$, The restriction $r \geq 1$ is essentially necessary in our main result. 

Finally, we use ideas from our hardness proof to provide a~\mbox{$(4k \log n, 1)$-bicriteria} approximation. 

Our results are summarised in Table~\ref{table:contributions}. For a technical overview of our results, see Section~\ref{section:technical_overview}. 

\renewcommand{\arraystretch}{1.4}
\setlength{\tabcolsep}{3mm}
\begin{table}[ht]
    \centering
    \begin{tabular}{||c|c|c||c||}
        \hline
        Sparsity ($f$)  & Dilation ($g$) & Complexity & Reference\\
        \hline
        \hline
        1 & $(1+\delta)(k+1)$ & $O(n^3 (\log n + \log \frac 1 \delta))$  
            & Gudmundsson and Wong~\cite{DBLP:journals/talg/GudmundssonW22} \\
        \hline
        $2 + \varepsilon$ & $O_\varepsilon((1+\delta)\log(k))$ & $O(n^3 (\log n + \log \frac 1 \delta))$ 
            & $r=O_\varepsilon(\log(k))$ in~Theorem~\ref{theorem:main_greedy_upper}\\
        \hline
        $4$ & $2(1+\delta)\log (2k)$ & $O(n^3 (\log n + \log \frac 1 \delta))$ 
            & $r=\log(2k)$ in~Theorem~\ref{theorem:main_greedy_upper}\\
        \hline
        $2^{1+\varepsilon} \ k^{\varepsilon}$ & 2$(1+\delta)\varepsilon^{-1}$ & $O(n^3 (\log n + \log \frac 1 \delta))$ 
            & $r=1/\varepsilon$ in~Theorem~\ref{theorem:main_greedy_upper}\\
        \hline
        $2\sqrt{2} \ \sqrt k$ & $4(1+\delta)$ & $O(n^3 (\log n + \log \frac 1 \delta))$ 
            & $r=2$ in~Theorem~\ref{theorem:main_greedy_upper}\\
        \hline
        $4k$ & $2(1+\delta)$ & $O(n^3 (\log n + \log \frac 1 \delta))$ 
            & $r=1$ in~Theorem~\ref{theorem:main_greedy_upper}\\
        \hline
        \multicolumn{3}{||c||} {Our analysis in Theorem~\ref{theorem:main_greedy_upper} is tight under EGC}
            & Theorem~\ref{theorem:main_greedy_lower}\\
        \hline
        $h(k)$ & $2 - \varepsilon$ & W[1]-hard
            & Theorem~\ref{theorem:main_set_cover_lower} \\
        \hline
        $4k \log n$ & 1 & $O(n^6 \log n)$ 
            & Theorem~\ref{theorem:main_set_cover_upper} \\
        \hline
    \end{tabular}
    \caption{The table shows the trade-off between sparsity~$f$ and dilation~$g$ in our bicriteria approximation algorithms for Problem~\ref{problem:bicriteria}. Note that EGC is the Erd\H{o}s girth conjecture, $h(\cdot)$ is any computable function, and $O_{\varepsilon}(\cdot)$ hides dependence on~$\varepsilon$.}
    \label{table:contributions}
\end{table}
\subsection{Related work}
\label{subsection:related_work}

Due to the difficult nature of Problem~\ref{problem:graph_augmentation}, most of the literature focuses on the special case where $k=1$. Farshi, Giannopoulos and Gudmundsson~\cite{DBLP:journals/siamcomp/FarshiGG08} provide an~$O(n^4)$ time algorithm, and an~$O(n^3)$ time \mbox{3-approximation} when~$k=1$. Wulff-Nilsen~\cite{DBLP:journals/comgeo/Wulff-Nilsen10} presents an~$O(n^3 \log n)$ time algorithm. Luo and Wulff-Nilsen~\cite{DBLP:conf/isaac/LuoW08} improves the space requirement to linear. Aronov et al.~\cite{DBLP:journals/josis/AronovBBJJKLLSS11} provide a nearly-linear time algorithm in the special case where the graph is a simple polygon and an interior point.

A variant of Problem~\ref{problem:graph_augmentation} is to add~$k$ edges to a graph to minimise the diameter instead of the dilation. Frati, Gaspers, Gudmundsson and Mathieson~\cite{DBLP:journals/algorithmica/FratiGGM15} provide a fixed parameter tractable 4-approximation for the problem. Bil{\`{o}}, Gual{\`{a}} and Proietti~\cite{DBLP:journals/tcs/BiloGP12} provide bicriteria approximability and inapproximability results. Several special cases have been studied. Demaine and Zadimoghaddam~\cite{DBLP:conf/swat/DemaineZ10} consider adding~$k$ edges of length~$\delta$, where~$\delta$ is small relative to the diameter.
Gro{\ss}e~et~al.~\cite{DBLP:journals/ijfcs/GrosseKSGS19} present nearly-linear time algorithms for adding one edge to either a path or a tree in order to minimise its diameter. Follow up papers improve the running time of the algorithm for paths~\cite{DBLP:journals/comgeo/Wang18} and for trees~\cite{DBLP:journals/tcs/Bilo22a,DBLP:journals/tcs/WangZ21}. Bil{\`{o}}, Gual{\`{a}}, Stefano Leucci and Sciarria~\cite{DBLP:conf/wads/BiloGLS23} extend the linear time algorithm to approximate the minimum diameter when $k>1$ edges are added to a tree.

Another variant is to add~$k$ edges to a graph to minimise the radius. Gudmundsson, Sha and Yao~\cite{DBLP:conf/isaac/GudmundssonSY21} provide a \mbox{3-approximation} for adding~$k$ edges to a graph to minimise its radius. The problem of adding one edge to minimise the radius of paths~\cite{DBLP:journals/comgeo/JohnsonW21,DBLP:journals/ijcga/wangzhao2020} and trees~\cite{DBLP:conf/wads/GudmundssonS21} has also been studied.

A problem closely related to Problem~\ref{problem:tree_spanner} is to compute minimum dilation graphs. In his Master's thesis, Mulzer~\cite{mulzer2004minimum} studies minimum dilation triangulations for the regular \mbox{$n$-gon}. Eppstein and Wortman~\cite{DBLP:conf/compgeom/EppsteinW05} provide a nearly-linear time algorithm to compute a minimum dilation star of a set of points. Giannopoulos, Knauer and Marx~\cite{giannopoulos2007minimum} prove that, given a set of points, it is \mbox{NP-hard} to compute a minimum dilation tour or a minimum dilation path. Aronov et al.~\cite{DBLP:journals/comgeo/AronovBCGHSV08} show that, given~$n$ points, one can construct a graph with~$n-1+k$ edges and dilation~$O(n/(k+1))$.

Our algorithm for Problem~\ref{problem:bicriteria} uses the greedy spanner, which is among the most extensively studied spanner constructions. In general metrics, the greedy \mbox{$(2k-1)$-spanner} has~$O(n^{1+1/k})$ edges~\cite{DBLP:journals/dcg/AlthoferDDJS93}. In $d$-dimensional Euclidean space, the greedy \mbox{$(1+\varepsilon)$-spanner} has~$O(n\varepsilon^{-d+1})$ edges~\cite{DBLP:conf/compgeom/DasHN93,DBLP:books/daglib/NarasimhanSmid}. The sparsity-dilation trade-off is (existentially) optimal in both cases~\cite{DBLP:journals/siamcomp/FiltserS20,DBLP:conf/focs/LeS19}. 

\section{Technical overview}
\label{section:technical_overview}

We divide our technical overview into six subsections. In Section~\ref{subsection:techview_previous}, we summarise the previous algorithm of Gudmundsson and Wong~\cite{DBLP:journals/talg/GudmundssonW22}. In Section~\ref{subsection:techview_greedy_upper}, we give an overview of our main result, that is, our \mbox{$(2 \sqrt[r]{2} \ k^{1/r},2r)$}-bicriteria approximation for all~$r \geq 1$. In Section~\ref{subsection:techview_greedy_lower}, we present the main ideas for proving our analysis is tight, assuming the Erd\H{o}s girth conjecture. In Section~\ref{subsection:techview_setcover_lower}, we summarise our proof that it is \mbox{W[1]-hard} to obtain a~\mbox{$(h(k), 2 - \varepsilon)$-bicriteria} approximation, for any computable function~$h$ and for any $\varepsilon > 0$. In Section~\ref{subsection:techview_setcover_upper}, we present a~\mbox{$(4k \log n, 1)$-bicriteria} approximation. In Section~\ref{subsection:techview_structure}, we summarise the structure of the remainder of the paper.

\subsection[Previous algorithm]{Previous algorithm of~\cite{DBLP:journals/talg/GudmundssonW22}}
\label{subsection:techview_previous}

Gudmundsson and Wong's~\cite{DBLP:journals/talg/GudmundssonW22} algorithm constructs the greedy spanner with a simple modification. The traditional greedy $t$-spanner takes as input a set of vertices, i.e. an empty graph, however, the modified greedy \mbox{$t$-spanner~\cite{DBLP:journals/talg/GudmundssonW22}} takes as input a set of vertices and edges, i.e. a non-empty graph.

The greedy $t$-spanner construction has two steps. First, all edges that are not in the initial graph are sorted by their length. Second, the edges are processed from shortest to longest. A processed edge~$uv$ is added if~$d_G(u,v) > t \cdot d_M(u,v)$, otherwise the edge~$uv$ is not added.

For Problem~\ref{problem:graph_augmentation}, Gudmundsson and Wong's~\cite{DBLP:journals/talg/GudmundssonW22} show, in their main lemma, that if the greedy \mbox{$t$-spanner} adds at least~$k+1$ edges, then~\mbox{$t \leq (k+1) \, t^*$}. Here,~$t^*$ is the minimum dilation if~$k$ edges are added to our graph.  Using this lemma, they then perform a binary search over a multiplicative \mbox{$(1+\delta)$-grid} for a~$t \in \mathbb R$ such that the greedy \mbox{$(1+\delta)\,t$-spanner} adds at most~$k$ edges, but the greedy \mbox{$t$-spanner} adds at least~$k+1$ edges. Therefore, $(1+\delta)\, t$ is a~$(1+\delta)(k+1)$-approximation of~$t^*$, since we can add $k$ edges to obtain a \mbox{$(1+\delta)\,t$-spanner} and~$(1+\delta)\,t \leq (1+\delta)(k+1) \, t^*$. 

Next, we briefly summarise the proof that if $k+1$ edges are added by the greedy algorithm, then $t \leq (k+1) \ t^*$. In Lemma~2 of~\cite{DBLP:journals/talg/GudmundssonW22}, the authors use the~$k+1$ greedy edges to construct a set of~$k+1$ vectors in  a~$k$-dimensional vector space. They define~$I$ to be a linearly dependent subset of the $k+1$ vectors. In Theorem~5 of~\cite{DBLP:journals/talg/GudmundssonW22}, the authors use the linear dependence property of~$I$ to prove that~$t \leq |I| \cdot t^*$. Since $|I| \leq k+1$, they obtain~$t \leq (k+1) \ t^*$. Unfortunately, the vector space approach of~\cite{DBLP:journals/talg/GudmundssonW22} fails extend to Problem~\ref{problem:bicriteria}, if the dilation factors~$g$ is sublinear in~$k$, even if the sparsity factor~$f$ is allowed to be polynomial in~$k$. 

\subsection{Greedy bicriteria approximation}
\label{subsection:techview_greedy_upper}

Our algorithm is the same as the one in~\cite{DBLP:journals/talg/GudmundssonW22}. Our difference lies in our analysis of the greedy \mbox{$t$-spanner}, in particular, in our main lemma. 

For Problem~\ref{problem:bicriteria}, we show, in our main lemma, that if the greedy~$t$-spanner adds at least~$fk+1$ edges, then $t \leq g t^*$. We will specify~$f$ and~$g$ later. Then, we apply the same binary search procedure to find a~$t\in \mathbb R$ where the greedy~$(1+\delta)\,t$-spanner adds at most~$fk$ edges, but the greedy~$t$-spanner adds at least~$fk+1$ edges. Then~$(1+\delta)\,t$ is an~$(f,(1+\delta) g)$-bicriteria approximation of~$t^*$.

Next, we briefly summarise our new proof that if $fk+1$ edges are added, then $t \leq gt^*$. In order to extend our analysis to sublinear dilation factors~$g$, we abandon the vector space approach of~\cite{DBLP:journals/talg/GudmundssonW22}. Our main idea is to construct an auxiliary graph, which we call the girth graph. The girth graph is an unweighted graph with~$2k$ vertices and $fk+1$ edges. Instead of defining~$I$ to be a linearly dependent subset, we define~$I$ to be the shortest cycle in the girth graph. We use a classical result in graph theory to choose the values~$f=2 \sqrt[r]{2} \ k^{1/r}$ and~$g=2r$, so that~$|I| \leq g$. In our final step, we use the cycle property of~$I$ to carefully prove~$t \leq |I| \cdot t^*$, which implies~$t \leq g t^*$. 

Putting this all together, we obtain Theorem~\ref{theorem:main_greedy_upper}. For a full proof, see Section~\ref{section:greedy_upper}.

\begin{restatable}{theorem}{maingreedyupper}
\label{theorem:main_greedy_upper}
For all $r \geq 1$, there is an $(f,(1+\delta)g)$-bicriteria approximation for Problem~\ref{problem:bicriteria} that runs in $O(n^3 (\log n + \log \frac 1 \delta))$ time, where 
$$f = 2 \sqrt[r]{2} \ k^{1/r} \quad\mbox{and}\quad g = 2r.$$
\end{restatable}

\subsection{Greedy analysis is tight}
\label{subsection:techview_greedy_lower}

Our analysis in Theorem~\ref{theorem:main_greedy_upper} is tight. This means one cannot obtain better bounds (up to constant factors) using the greedy spanner. Our proof assumes the Erd\H{o}s girth conjecture~\cite{erdos1965some}. 

The girth of an unweighted graph is defined as the number of edges in its shortest cycle. In the proof of Theorem~\ref{theorem:main_greedy_upper}, we cite a classical result stating that a graph with~$n$ vertices and at least~$n^{1+1/r}$ edges has girth at most~$2r$. The Erd\H{o}s girth conjecture states that there are graphs with~$n$ vertices, at least $\Omega(n^{1+1/r})$ edges and girth $2r+2$. Several conditional lower bounds have been shown under the Erd\H{o}s girth conjecture, namely, the sparsity-dilation trade-off of the greedy spanner~\cite{DBLP:journals/dcg/AlthoferDDJS93}, and the space requirement of approximate distance oracles~\cite{DBLP:journals/jacm/ThorupZ05}. 

We summarise our construction that proves that our analysis is tight. Assuming the Erd\H{o}s girth conjecture, there exists a graph~$H$ with~$n=k+1$ vertices, $m = \Omega(n^{1+1/r})$ edges, and girth~$2r+2$. We construct a graph~$G$ so that if we run the algorithm in Theorem~\ref{theorem:main_greedy_upper}, the girth graph of~$G$ would be~$H$. We use the properties of~$H$ to show that, if there are~$k$ edges that can be added to~$G$ so that the resulting dilation is~$t^*$, then if we add~$m-1$ edges to~$G$ using the greedy \mbox{$t$-spanner} construction, the resulting dilation is at least~$(2r+1)\ t^*$.

Putting this all together, we obtain Theorem~\ref{theorem:main_greedy_lower}. For a full proof, see . 

\begin{restatable}{theorem}{maingreedylower}
\label{theorem:main_greedy_lower}
For all $r \geq 1$, assuming the Erd\H{o}s girth conjecture, there is a graph class for which the algorithm in Theorem~\ref{theorem:main_greedy_upper} returns an $(f,g)$-bicriteria approximation, where
$$f = \Omega(k^{1/r}) \quad\mbox{and}\quad g = 2r+1.$$
\end{restatable}

\subsection{Set cover reduction}
\label{subsection:techview_setcover_lower}

Next, we show that the restriction~$r \geq 1$ is necessary in Theorem~\ref{theorem:main_greedy_upper}. Recall that Theorem~\ref{theorem:main_greedy_upper} states that there is a~\mbox{$(2 \sqrt[r]{2} \ k^{1/r},(1+\delta)\, 2r)$-bicriteria} approximation algorithm for all~$r \geq 1$. We prove that it is \mbox{W[1]-hard} to obtain a $(h(k),2-\varepsilon)$-bicriteria approximation for any computable function~$h$ and for any~$\varepsilon > 0$. Our proof is a reduction from set cover.

We summarise our construction of the Problem~\ref{problem:bicriteria} instance. We show that every set cover instance can be reduced to a Problem~\ref{problem:bicriteria} instance. We represent each element with a pair of points, and we represent each set with a triple of points. In our Problem~\ref{problem:bicriteria} instance, we add edges to connect either the pairs or the triples. We show via an exchange argument that we only need to consider adding edges that connect triple. Connecting a triple corresponds to choosing a set, which lowers the dilation of all elements in that set to below the threshold value. Finally, we show that a~$(h(k), 2-\varepsilon)$-bicriteria approximation for our Problem~\ref{problem:bicriteria} would solve set cover within an approximation factor of $h(k)$. However, it is \mbox{W[1]-hard} to obtain an $h(k)$-approximation algorithm for any computatable function~$h$~\cite{DBLP:journals/jacm/SLM19}, see Fact~\ref{fact:fpt_hardness}.

Putting this all together, we obtain Theorem~\ref{theorem:main_set_cover_lower}. For a full proof, see the full version of this paper, see Section~\ref{sec:set_cover_lower}. 

\begin{restatable}{theorem}{mainsetcoverlower}
\label{theorem:main_set_cover_lower}
For all $\varepsilon > 0$, assuming FPT $\neq$ W[1], one cannot obtain an $(f,g)$-bicriteria approximation for Problem~\ref{problem:bicriteria}, where
$$f = h(k) \quad\mbox{and}\quad g = 2 - \varepsilon,$$
and $h(\cdot)$ is any computable function.
\end{restatable}

\subsection{Set cover algorithm}
\label{subsection:techview_setcover_upper}

Finally, we provide a $(4k \log n, 1)$-bicriteria approximation that runs in~$O(n^6 \log n)$ time. Our main idea is to formulate the problem into a set cover instance, and then to apply an~$O(\log n)$-approximation algorithm for set cover~\cite{DBLP:journals/mor/Chvatal79}.

We state our algorithm. For each~$t \in \mathbb R$, we define a set cover instance $\mathcal I_t$. We construct the set cover instance $\mathcal I_t$ so that its elements are defined by pairs of vertices in $V(G)$, and each set in $\mathcal I_t$ is associated with a pair of vertices in $V(G)$. Formally, the elements of $\mathcal I_t$ are $\{(u,v): u,v \in V(G)\}$. The sets of $\mathcal I_t$ are $\{S_e: e \in V(G) \times V(G)\}$, where each set~$S_e$ contains all pairs $(u,v)$ that have dilation at most~$t$ after the edge~$e$ is added to~$G$. To formally define $S_e$, let~$G_e$ be the graph if an edge~$e$ is added to~$G$, and define $S_e = \{(u,v): d_{G_e}(u,v) \leq t \cdot d_M(u,v)\}$. Now, we can apply the algorithm of~\cite{DBLP:journals/mor/Chvatal79} on~$\mathcal I_t$ to obtain a set cover~$\mathcal S$. If $|\mathcal S| < k$, then we can add fewer than $k$ edges to $G$ to reduce its dilation to $t$, so $t > t^*$. We claim that if $|\mathcal S| > 4k^2 \log n$, then $t \leq t^*$. Therefore, we perform binary search on~$t$ in the same way as~\cite{DBLP:journals/talg/GudmundssonW22} to obtain a $(4k \log n,1)$-bicriteria approximation.

To prove correctness, it remains to show our claim that $|\mathcal S| > 4k^2 \log n  \implies t \leq t^*$. We show the contrapositive. If $t > t^*$, from the definition of~$t^*$ there exists~$k$ edges that can be added to~$G$ to make it a~$t$-spanner. Consider a clique with vertices that are the endpoints of the~$k$ edges. Adding these~$2k^2$ edges to the graph would make it a~$t$-spanner. Moreover, each $t$-path, that is, a~$(u,v)$-path with length at most~$t \cdot d_M(u,v)$, uses at most one edge in the clique. Therefore, the union of the sets~$S_e$, where~$e$ is an edge in the clique, forms a set cover over all pairs of vertices~$(u,v)$. The optimal solution of the set cover instance is at most~$2k^2$. The algorithm of~\cite{DBLP:journals/mor/Chvatal79} returns an $O(\log (n^2))$-approximation, since the number of elements and sets in $\mathcal I_t$ is $O(n^2)$. Putting this together, our algorithm returns a set cover $\mathcal S$ such that $|\mathcal S| \leq 2k^2 \log (n^2) = 4k^2 \log n$, as required.

Next, we analyse the running time. Computing $S_e$ takes $O(n^3)$ time for each $e \in V(G) \times V(G)$. Therefore, constructing the set cover instance $\mathcal I_t$ takes $O(n^5)$ time. The number of elements and the number of sets in $\mathcal I_t$ is $O(n^2)$. Therefore, the cubic time algorithm of~\cite{DBLP:journals/mor/Chvatal79} takes $O(n^6)$ time in total. Finally, performing the $\log n)$ steps in the search brings the total running time to $O(n^6 \log n)$.

Putting this all together, we obtain Theorem~\ref{theorem:main_set_cover_upper}.

\begin{theorem}
\label{theorem:main_set_cover_upper}
There is an $(f,g)$-bicriteria approximation for Problem~\ref{problem:bicriteria}, where
$$f = 4k \log n \quad\mbox{and}\quad g = 1.$$
\end{theorem}

This completes the overview of the main results of this paper. 

\subsection{Structure of paper}
\label{subsection:techview_structure}

The structure of the remainder of our paper is summarised in Table~\ref{table:structure}.

\renewcommand{\arraystretch}{1.3}
\setlength{\tabcolsep}{5mm}
\begin{table}[ht]
    \centering
    \begin{tabular}{||c|c|c||}
        \hline
         & Reference & Proof\\
        \hline
        \hline
        $(2 \sqrt[r]{2} \ k^{1/r},(1+\delta)\, 2r)$-bicriteria approximation & 
        Theorem~\ref{theorem:main_greedy_upper} & 
        Section~\ref{section:greedy_upper}\\
        \hline
        {Theorem~\ref{theorem:main_greedy_upper} analysis is tight} &
        Theorem~\ref{theorem:main_greedy_lower} &
        Section~\ref{section:greedy_lower}
        \\
        \hline
        $(h(k),2 - \varepsilon)$-bicriteria approximation is W[1]-hard &
        Theorem~\ref{theorem:main_set_cover_lower} &
        Section~\ref{sec:set_cover_lower}
        \\
        \hline
    \end{tabular}
    \caption{Proofs of Theorems~\ref{theorem:main_greedy_upper}, \ref{theorem:main_greedy_lower}, \ref{theorem:main_set_cover_lower} can be found in Section~\ref{section:greedy_upper}, Section~\ref{section:greedy_lower}, and Section~\ref{sec:set_cover_lower}.}
    \label{table:structure}
\end{table}

\section{Greedy bicriteria approximation}
\label{section:greedy_upper}

In this section, we will prove Theorem~\ref{theorem:main_greedy_upper}. We restate the theorem for convenience.

\maingreedyupper*

Recall from Section~\ref{section:introduction} that the vertices and edges of~$G$ are~$V(G)$ and~$E(G)$ respectively. Let~$e \in V(G) \times V(G)$ be an edge not necessarily in~$E(G)$. Let~$d_M(e)$ denote the length of the edge~$e$ in the metric space~$M$ and let~$d_G(e)$ denote the shortest path distance between the endpoints of~$e$ in the graph~$G$. Consider a minimum dilation graph~$G^*$ after adding an optimal set~$S^*$ of~$k$ edges to~$G$. Let~$t^*$ be the dilation of~$G^*$.

Recall from Section~\ref{section:technical_overview} that our approach is to use the greedy~$t$-spanner construction. We formalise the construction in the definition below.

\begin{definition}
\label{definition:greedy_t_spanner}
Define $G_0 = G$, and for~$i \geq 1$, define~$G_i = G_{i-1} \cup \{a_i\}$, where $a_i$ is the shortest edge in $V(G) \times V(G)$ satisfying $d_{G_{i-1}}(a_i) > t \cdot d_M(a_i)$. The process halts if no edge~$a_i$ exists.
\end{definition}

We have two cases: either the process halts after adding more than~$fk$ edges, or after adding at most~$fk$ edges. If more than~$fk$ edges are added, we show a dilation bound on~$t$. In particular, Lemma~\ref{lemma:t_gt} states that if there is an edge~$a_i$ satisfying $d_{G_{i-1}}(a_i) > t \cdot d_M(a_i)$ for all~$1 \leq i \leq fk+1$, then we have the dilation bound~$t \leq gt^*$. We will specify the parameters~$f, g \geq 1$ later in this section.

Our approach is to construct an auxilliary graph~$H$, which we will also refer to as the girth graph. Define the vertices of~$H$ to be~$V(H) = \{v_1, \ldots, v_{2k}\}$. Each vertex in~$V(H)$ corresponds to an endpoint of an edge in the optimal set of~$k$ edges~$S^*$. In particular, let $S^* = \{s_1, \ldots, s_k\}$, and let~$v_{2i-1},v_{2i} \in V(H)$ correspond to the endpoints of~$s_i$. Define the edges of~$H$ to be~$E(H) = \{e_1, \ldots, e_{fk+1}\}$. We will describe the procedure for constructing each edge~$e_i$. 

Consider the greedy edge~$a_i$, see Figure~\ref{fig:3_greedy_upper_bound_01}. Define~$\delta_{G^*}(a_i)$ to be the shortest path between the endpoints of~$a_i$ in~$G^*$, shown in grey in Figure~\ref{fig:3_greedy_upper_bound_01}. Note that~$\delta_{G^*}(a_i)$ denotes a path, whereas~$d_{G^*}(a_i)$ denotes a length. Suppose that there are no edges in~$S^*$ along the path~$\delta_{G^*}(a_i)$, for some~$1 \leq i \leq fk+1$. Then, 
\[
    t^* \cdot d_M(a_i) \geq d_{G^*}(a_i) = d_G(a_i) \geq d_{G_{i-1}}(a_i) > t \cdot d_M(a_i),
\]
so $t < t^* \leq gt^*$, which would already imply Lemma~\ref{lemma:t_gt}. 

\begin{figure}[ht]
    \centering
    \includegraphics{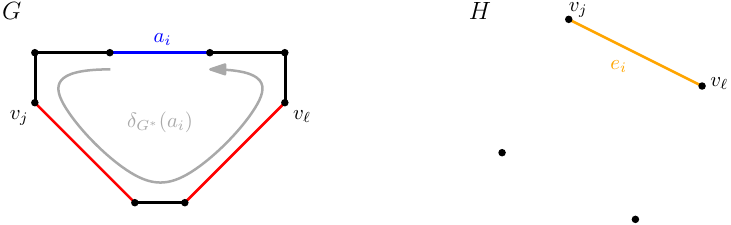}
    \caption{Left: The graph~$G$ (black), the greedy edge~$a_i$ (blue), the path~$\delta_{G^*}(a_i)$ (grey), and the edges~$\delta_{G^*}(a_i) \cap S^*$ (red). Right: The girth graph~$H$ and the edge~$e_i$ (orange).}
    \label{fig:3_greedy_upper_bound_01}
\end{figure}

Therefore, we can assume that~$\delta_{G^*}(a_i)$ contains at least one edge in~$S^*$, for every $i = 1, \ldots, fk+1$. Consider the edges~$\delta_{G^*}(a_i) \cap S^*$, shown in red in Figure~\ref{fig:3_greedy_upper_bound_01}. Choose a direction for the path~$\delta_{G^*}(a_i)$, sort the list of endpoints of~$\delta_{G^*}(a_i) \cap S^*$ with respect to this direction, and let the first and last endpoints in the sorted list be~$v_j$ and $v_\ell$. Another way to characterise~$v_j$ (respectively~$v_\ell$) is that~$v_j$ is an endpoint of an edge in~$\delta_{G^*}(a_i) \cap S^*$ so that the shortest path between~$v_j$ and one of the endpoints of~$a_i$ contains no edges in~$S^*$ (respectively the other endpoint of~$a_i$). Finally, we define~$e_i$ to be the edge in~$H$ connecting~$v_j$ to~$v_\ell$. Note that~$e_i$ is an undirected, unweighted edge, shown in orange in Figure~\ref{fig:3_greedy_upper_bound_01}. This completes the construction of~$H$. 

In Figure~\ref{fig:3_greedy_upper_bound_02}, we provide a more complete example of a graph~$G$ and its girth graph~$H$. The optimal set of~$k=4$ edges is~$S^* = \{s_1, s_2, s_3, s_4\}$, which is shown in red. The five greedy edges $\{a_1, a_2, a_3, a_4, a_5\}$ are shown in blue. The first and last endpoints of~$\delta_{G^*}(a_1) \cap S^*$ are~$v_1$ and~$v_3$, so~$e_1 = v_1v_3$. Similarly,~$e_2 = v_3v_5$, $e_3 = v_5v_7$, $e_4 = v_1 v_7$ and~$e_5 = v_5v_6$. 

\begin{figure}[ht]
    \centering
    \includegraphics{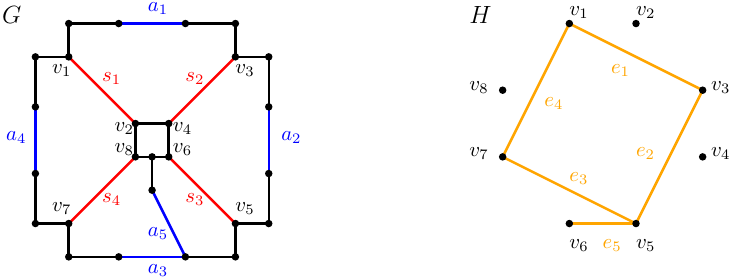}
    \caption{Left: The graph~$G$ (black), the optimal edges~$s_1, \ldots, s_4$ (red), and the greedy edges~$a_1, \ldots, a_5$ (blue). Right: The girth graph~$H$ has edges $e_1, \ldots, e_5$ (orange) and a girth of~$4$.}
    \label{fig:3_greedy_upper_bound_02}
\end{figure}

Next, define~$J$ to be the shortest cycle in $H$, and define~$I = \{j: e_j \in J\}$. Therefore, the girth of~$H$ is $|J| = |I|$. Note that $H$ has $2k$ vertices and $fk+1$ edges, so $J$ is guaranteed to exist if $f \geq 2$.

We use a classical result in graph theory to set the parameters~$f$ and~$g$.

\begin{lemma}
\label{lemma:girth_lemma}
A graph with~$n$ vertices and at least~$n^{1 + 1/r}+1$ edges has girth at most~$2r$. 
\end{lemma}

\begin{proof}
The lemma is a classical result~\cite{bollobas2004extremal}. Lemma 2 of~\cite{leucci2019graph} provides a self-contained proof.
\end{proof}

With Lemma~\ref{lemma:girth_lemma} in mind, we set~$f = 2 \sqrt[r]{2} \ k^{1/r}$ and $g = 2r$, where~$r \geq 1$. Then, the graph~$H$ has $2k$ vertices, $(2k)^{1+1/r}+1$ edges, and therefore~$H$ has girth~$|I| \leq g = 2r$. Having defined the girth graph~$H$, the indices~$I$, and the parameters~$f$ and~$g$, the next step is to prove Lemma~\ref{lemma:t_gt}.

\begin{restatable}{lemma}{tgt}
\label{lemma:t_gt}
If~$a_j$ exists for all~$j = 1,\ldots,fk+1$, then~$t \leq gt^*$.
\end{restatable}

We divide the proof of Lemma~\ref{lemma:t_gt} into three lemmas. In Lemma~\ref{lemma:path}, we construct a path. In Lemma~\ref{lemma:path_lower_bound}, we lower bound the length of the path. In Lemma~\ref{lemma:path_upper_bound}, we upper bound the length of the path. We start defining the path. 

\begin{lemma}
\label{lemma:path}
Let~$i = \max I$. There is a path in~$G$ between the endpoints of~$a_i$ using only edges in 
\[
    \{G \cap \delta_{G^*}(a_j): j \in I\} \cup \{a_j: j \in I \setminus \{i\}\}.
\]
\end{lemma}

\begin{proof}
Recall that~$J = \{e_j: j \in I\}$ is a cycle in~$H$. After removing the edge~$e_i$, there is still a path in~$J \subseteq H$ between the endpoints of~$e_i$. Let the vertices along this path be~$w_1, \ldots, w_m$, where~$e_i = w_1w_m$, and~$w_\ell w_{\ell+1} \in J$ for all~$\ell = 1,\ldots,m-1$. In Figure~\ref{fig:3_greedy_upper_bound_04}, the path~$w_1, \ldots, w_m$ is shown in orange. Let the endpoints of~$a_i$ be~$a_i(0)$ and~$a_i(1)$. Our approach is to use the path~$w_1, \ldots, w_m \subset H$ to construct a path between~$a_i(0)$ and~$a_i(1)$ that only uses edges in~$\{G \cap \delta_{G^*}(a_j): j \in I\} \cup \{a_j: j \in I \setminus \{i\}\}$. 

\begin{figure}[ht]
    \centering
    \includegraphics{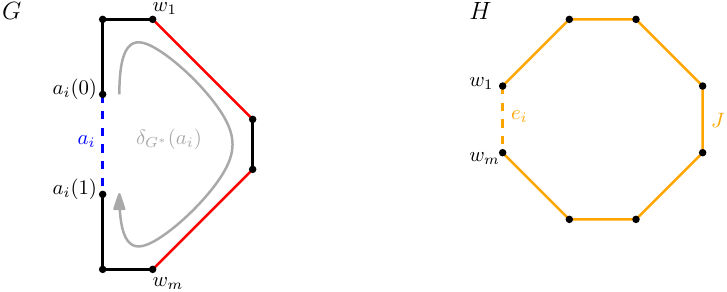}
    \caption{Left: The graph~$G$ (black), the greedy edge~$a_i$ (blue), the path~$\delta_{G^*}(a_i)$ (grey), and the edges~$\delta_{G^*}(a_i) \cap S^*$ (red). Right: The girth graph~$H$, the cycle $J$ (orange), and edge~$e_i$ (dashed).}
    \label{fig:3_greedy_upper_bound_04}
\end{figure}

First, we consider the edge~$e_i = w_1 w_m$. Recall from the definition of~$V(H)$ that~$w_1$ and~$w_m$ are endpoints of edges in the optimal set~$S^*$. Moreover, from the definition of $e_i \in E(H)$, we know that~$w_1$ and~$w_m$ are the first and last endpoints of~$S^*$ along the path~$\delta_{G^*}(a_i)$. The path~$\delta_{G^*}(a_i)$ is shown in grey in Figure~\ref{fig:3_greedy_upper_bound_04}. The endpoints of~$\delta_{G^*}(a_i)$ are~$a_i(0)$ and~$a_i(1)$. Therefore, the subpath of~$\delta_{G^*}(a_i)$ between~$a_i(0)$ and~$w_1$ only uses edges in~$G$ and no edges in~$S^* = G^* \setminus G$. Therefore, the subpath only uses edges in~$G \cap \delta_{G^*}(a_i)$. The subpath from~$a_i(0)$ to~$w_1$ is shown in black in Figure~\ref{fig:3_greedy_upper_bound_04}. Similarly, there is a path between~$w_m$ and~$a_i(1)$ using only edges in~$G \cap \delta_{G^*}(a_i)$.

Next, we consider the edge~$e_j = w_\ell w_{\ell+1}$, where $1 \leq \ell \leq m-1$,~$j \in I$ and~$j < i$. Let the endpoints of~$a_j$ be~$a_j(0)$ and~$a_j(1)$. From the definition of~$e_j = w_\ell w_{\ell+1}$, there is a subpath of~$\delta_{G^*}(a_j)$ between~$w_\ell$ and~$a_j(0)$ that only uses edges in~$G \cap \delta_{G^*}(a_j)$. Similarly, there is subpath of~$\delta_{G^*}(a_j)$ between~$a_j(1)$ and~$w_{\ell+1}$ that only uses edges in~$G \cap \delta_{G^*}(a_j)$. Therefore, there is a path between~$w_\ell$ and~$w_{\ell+1}$ that uses only edges in~$\{G \cap \delta_{G^*}(a_j)\} \cup a_j$.

See Figure~\ref{fig:3_greedy_upper_bound_03} for an example. Consider the edge~$e_1 = w_1 w_2$. There is a path between~$w_1$ and~$w_2$ that only uses edges in~$\{G \cap \delta_{G^*}(a_j)\}$, which are black edges, and the blue edge~$a_1$. Similarly arguments apply for~$e_2 = w_2 w_3$ and~$e_3 = w_3 w_4$.

\begin{figure}[ht]
    \centering
    \includegraphics{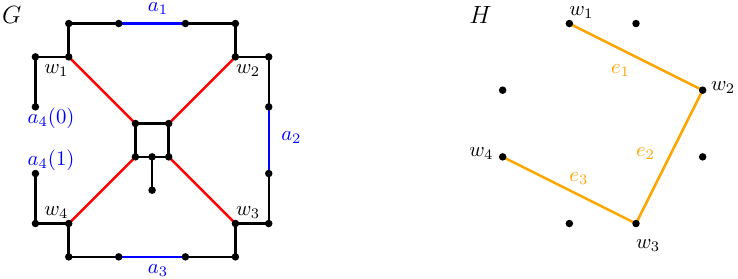}
    \caption{The path $w_1,w_2,w_3,w_4$ is shown on the right. There is a path between~$a_4(0)$ and~$a_4(1)$ only using the blue edges~$a_1,a_2,a_3$ and black edges in~$\delta_{G^*}(a_1),\delta_{G^*}(a_2),\delta_{G^*}(a_3)$ or~$\delta_{G^*}(a_4)$.}
    \label{fig:3_greedy_upper_bound_03}
\end{figure}

The final step is to put it all together. There is a path between~$a_i(0)$ and~$w_1$ that only uses edges in~$G \cap \delta_{G^*}(a_i)$. For~$\ell = 1,\ldots,m-1$, there is a path between~$w_\ell$ and~$w_{\ell+1}$ that only uses edges in~$\{G \cap \delta_{G^*}(a_j)\} \cup a_j$, where~$j \in I \setminus \{i\}$. There is a path between~$w_m$ and~$a_i(1)$ that only uses edges in~$G \cap \delta_{G^*}(a_i)$. Therefore, there is a path between~$a_i(0)$ and~$a_i(1)$ that only uses edges in~$\{G \cap \delta_{G^*}(a_j): j \in I\} \cup \{a_j: j \in I \setminus \{i\}\}$, as required.
\end{proof}

In Lemma~\ref{lemma:path_lower_bound}, we show a lower bound on the length of the path in Lemma~\ref{lemma:path}.

\begin{lemma}
\label{lemma:path_lower_bound}
The length of the path in Lemma~\ref{lemma:path} is at least~$t \cdot d_M(a_i)$. 
\end{lemma}

\begin{proof}
From Definition~\ref{definition:greedy_t_spanner}, we have $d_{G_{i-1}}(a_i) > t \cdot d_M(a_i)$. Therefore, any path in~$G_{i-1}$ between the endpoints of~$a_i$ has length at least~$t \cdot d_M(a_i)$. It suffices to show that the path is in~$G_{i-1}$. By Lemma~\ref{lemma:path}, all of the edges in the path are in $\{G \cap \delta_{G^*}(a_j): j \in I\}$ or $\{a_j: j \in I \setminus \{i\}\}$. But $\{G \cap \delta_{G^*}(a_j): j \in I\} \subseteq G \subseteq G_{i-1}$ and $\{a_j: j \in I \setminus \{i\}\} \subseteq G_{i-1}$. So the path is in~$G_{i-1}$ and its length is at least~$t \cdot d_M(a_i)$.
\end{proof}

In Lemma~\ref{lemma:path_upper_bound}, we upper bound the length of the path in Lemma~\ref{lemma:path}. 

\begin{lemma}
\label{lemma:path_upper_bound}    
If $t > gt^*$, then the length of the path in Lemma~\ref{lemma:path} is at most~$|I| \cdot t^* \cdot d_M(a_i)$.
\end{lemma}

\begin{proof}
Given a set of edges~$E$, let~$\total(E)$ denote the total sum of edge lengths in~$E$. Recall that the path in Lemma~\ref{lemma:path} only uses edges in~$\{G \cap \delta_{G^*}(a_j): j \in I\} \cup \{a_j: j \in I \setminus \{i\}\}$. A na\"ive approach to prove the lemma is to bound~$\total(\{\delta_{G^*}(a_j): j \in I\} \cup \{a_j: j \in I \setminus \{i\}\})$. Note that~$G$ is removed from the first set of braces. We have 
\[
\begin{array}{rcl}
\total(\{\delta_{G^*}(a_j): j \in I\}) &\leq& \sum_{j \in I} t^* \cdot d_M(a_j) \leq |I| \cdot t^* \cdot d_M(a_i),\\
\total(\{a_j: j \in I\setminus \{i\}\}) &=& \sum_{j \in I\setminus \{i\}} d_M(a_j) \leq (|I|-1) \cdot d_M(a_i).\\
\end{array}
\]
Therefore, the total length of the path is at most~$(|I| \cdot t^* + |I| - 1) \cdot d_M(a_i) < |I| \cdot (t^* + 1) \cdot d_M(a_i)$. Since $(t^* + 1) \leq 2t^*$, we have proven Lemma~\ref{lemma:path_upper_bound} up to a factor of~$2$. This analysis would already yield an~$(f,2g)$-bicriteria approximation. However, to shave off the factor of~2 and obtain a tight analysis, we need a more sophisticated argument.

We strengthen our upper bound by re-introducing~$G$ back into the first set of braces, in other words, by bounding~$\total(\{G \cap \delta_{G^*}(a_j)\})$. Since~$G = G^* \setminus S^*$, we write
\[
\total(\{G \cap \delta_{G^*}(a_j)\}) = \total(\{\delta_{G^*}(a_j)\}) - \total(\{S^* \cap \delta_{G^*}(a_j)\}).
\]
We have two cases, depending on the size of~$\total(\{S^* \cap \delta_{G^*}(a_j)\})$.

\textbf{Case 1.}~$\total(\{S^* \cap \delta_{G^*}(a_j)\}) < (1 - \frac {1}{|I|}) \cdot d_M(a_j)$ for some~$j \in I$. Then for every~$s \in \{S^* \cap \delta_{G^*}(a_j)\}$, it holds that~$d_M(s) < d_M(a_i)$. Therefore,~$d_{G_{j-1}}(s) \leq t \cdot d_M(s)$, since~$a_j$ is the shortest edge in~$G_{j-1}$ satisfying~$d_{G_{j-1}}(a_j) > t \cdot d_M(a_j)$. Let the endpoints of~$a_j$ be~$a_j(0)$ and~$a_j(1)$. Let the edges of~$\delta_{G^*}(a_j) \cap S^*$ be~$s_1,\ldots,s_m$, and let the endpoints of~$s_i$ be~$w_{2i-1}$ and~$w_{2i}$. Assume without loss of generality that the endpoints~$w_1,\ldots,w_{2m}$ are in sorted order along the path~$\delta_{G^*}(a_j)$. Then,
\[
\begin{array}{rcl}
d_{G_{j-1}}(a_j) 
&\leq& d_{G_{j-1}}(a_j(0),w_1) + \sum_{i=1}^m d_{G_{j-1}}(w_{2i-1},w_{2i}) + \sum_{i=1}^{m-1} d_{G_{j-1}}(w_{2i},w_{2i+1}) \\&& \quad + d_{G_{j-1}}(w_{2m},a_j(1))
\\
&=& d_{G_{j-1}}(a_j(0),w_1) + \sum_{i=1}^m d_{G_{j-1}}(s_i) + \sum_{i=1}^{m-1} d_{G_{j-1}}(w_{2i},w_{2i+1}) \\&& \quad + d_{G_{j-1}}(w_{2m},a_j(1))
\\
&\leq& d_{G_{j-1}}(a_j(0),w_1) + \sum_{i=1}^{m-1} d_{G_{j-1}}(w_{2i},w_{2i+1}) + d_{G_{j-1}}(w_{2m},a_j(1)) \\&& \quad + \sum_{i=1}^m t \cdot d_M(s_i)
\\
&\leq& d_{G_{}}(a_j(0),w_1) + \sum_{i=1}^{m-1} d_{G_{}}(w_{2i},w_{2i+1}) + d_{G_{}}(w_{2m},a_j(1)) \\&& \quad + \sum_{i=1}^m t \cdot d_M(s_i)
\\
&=& d_{G^*}(a_j(0),w_1) + \sum_{i=1}^{m-1} d_{G^*}(w_{2i},w_{2i+1}) + d_{G^*}(w_{2m},a_j(1)) \\&& \quad + \sum_{i=1}^m t \cdot d_M(s_i)
\\
&<& d_{G^*}(a_j) + \sum_{i=1}^m t \cdot d_M(s_i),
\\
\end{array}
\]
where the first line uses the triangle inequality, the second line uses~$s_i = w_{2i-1}w_{2i}$, the third line uses~$d_{G_{j-1}}(s) \leq t \cdot d_M(s)$, the fourth line uses~$G \subset G_{i-1}$, the fifth line uses that all the subpaths no longer use edges in~$S^*$, and the sixth line uses that all edges are a subset of the edges in~$\delta_G^*(a_j)$. Therefore,
\[
\begin{array}{rcl}
t \cdot d_M(a_j) < d_{G_{j-1}}(a_j) 
&\leq& d_{G^*}(a_j) + \sum_{i=1}^m t \cdot d_M(s_i) 
\\
&\leq& t^* \cdot d_M(a_j) + t \cdot \total(\{S^* \cap \delta_{G^*}(a_j)\})
\\
&=& t^* \cdot d_M(a_j) + t \cdot (1 - \frac {1}{|I|}) \cdot d_M(a_j).
\end{array}
\]
Simplifying, we get $t < t^* + t - \frac t {|I|}$, which implies $t < |I| \cdot t^* = gt^*$. But this contradicts $t > gt^*$ in the lemma statement. Therefore, only Case 2 remains.

\textbf{Case 2.}~$\total(\{S^* \cap \delta_{G^*}(a_j)\}) \geq (1 - \frac {1}{|I|}) \cdot d_M(a_j)$ for all~$j \in I$. Let $L$ be the length of the path in Lemma~\ref{lemma:path}. Then,
\[
\begin{array}{rcl}
L &\leq& \total(\{G \cap \delta_{G^*}(a_j): j \in I\}) + \total(\{a_j: j \in I \setminus \{i\}\})
\\
&=& \total(\{\delta_{G^*}(a_j):j \in I\}) - \total(\{S^* \cap \delta_{G^*}(a_j): j \in I\}) + \sum_{j \in I\setminus \{i\}} d_M(a_j)
\\
&\leq& \sum_{j \in I} t^* \cdot d_M(a_j) - \sum_{j \in I} (1 - \frac {1}{|I|}) \cdot d_M(a_j) + \sum_{j \in I\setminus \{i\}} d_M(a_j)
\\
&=& \sum_{j \in I} t^* \cdot d_M(a_j) - (1-\frac 1 {|I|}) \cdot d_M(a_i) + \sum_{j \in I\setminus \{i\}} \frac 1 {|I|} \cdot d_M(a_j)
\\
&\leq& |I| \cdot t^* \cdot d_M(a_i) - (1-\frac 1 {|I|}) \cdot d_M(a_i) + (\frac {|I| - 1} {|I|}) \cdot d_M(a_i) 
\\
&=& |I| \cdot t^* \cdot d_M(a_i),
\end{array}
\]
where the first line uses Lemma~\ref{lemma:path}, the second line uses~$G = G^* \setminus S^*$, the third line uses the assumption from the case distinction, and fourth, fifth and sixth lines simplify the expression. Therefore,~$L \leq |I| \cdot t^* \cdot d_M(a_i)$, as required.
\end{proof}

Combining Lemmas~\ref{lemma:path}-\ref{lemma:path_upper_bound}, we obtain Lemma~\ref{lemma:t_gt}, which we will restate for convenience.

\tgt*

Finally, we use Lemma~\ref{lemma:t_gt} to prove Theorem~\ref{theorem:main_greedy_upper}. The idea is to combine the sparsity bound in the case where the greedy construction halts after adding at most~$fk$ edges, with the dilation bound~$t \leq g t^*$ in the case where the greedy construction halts after adding at least~$fk+1$ edges.

\maingreedyupper*

\begin{proof}
First, we describe the decision algorithm. Given any~$t \in \mathbb R$, the decision algorithm is to construct the greedy~$t$-spanner as described in Definition~\ref{definition:greedy_t_spanner}. If at most~$fk$ edges are added, then we continue searching over dilation values that are less than~$t$. If at least~$fk+1$ edges are added, then we continue searching over dilation values that are greater than~$t$.

Second, we perform a binary search to obtain an~$(f,(1+\delta)g)$-bicriteria approximation for Problem~\ref{problem:bicriteria}. Given a set of vertices, Gudmundsson and Wong~\cite{DBLP:journals/talg/GudmundssonW22} show how to (implicitly) binary search a set of~$O(n^4)$ critical values, so that the dilation of any graph with those vertices will be within a factor of~$O(n)$ of one of the critical values. We refine the search to a multiplicative~$(1+\delta)$-grid. As a result, we obtain a~$t \in \mathbb R$ where a greedy $t$-spanner adds at least~$fk+1$ edges, but a greedy $(1+\delta)t$-spanner adds at most~$fk$ edges. By Lemma~\ref{lemma:t_gt}, we have~$t \leq gt^*$, so~$(1+\delta) t \leq (1+\delta) gt^*$. The greedy~$(1+\delta)t$-spanner adds at most~$fk$ edges to the graph and its dilation is at most~$(1+\delta) gt^*$, so we have an~$(f,(1+\delta)g)$-bicriteria-approximation.

Third, we analyse the running time. The running time of the decision algorithm is~$O(n^3)$~\cite{DBLP:journals/talg/GudmundssonW22}. We perform the binary search by first calling the decider~$O(\log n)$ times on the critical values, and an additional~$O(\log \frac 1 \delta)$ times on the multiplicative~$(1+\delta)$-grid. 
\end{proof}

\section{Greedy analysis is tight}
\label{section:greedy_lower}

In this section, we will prove Theorem~\ref{theorem:main_greedy_lower}. We restate the theorem for convenience.

\maingreedylower*

Recall that the Erd\H{o}s girth conjecture~\cite{erdos1965some} states that, for all positive integers~$n$ and~$r$, there exists a graph~$H$ with~$n$ vertices, $m= \Omega(n^{1+1/r})$ edges, and girth~$2r+2$. Recall that the girth of a graph is the length of its shortest cycle. Note that Theorem~\ref{theorem:main_greedy_lower} does not contradict Theorem~\ref{theorem:main_greedy_upper}, as the constant in~$\Omega(k^{1/r})$ is less than~$2 \sqrt[r]{2}$. One would need to resolve the Erd\H{o}s girth conjecture to determine the precise constant.

We summarise our approach. We construct a graph~$G$ so that its girth graph is~$H$. Using the properties of the girth graph, we show that the greedy spanner gives an $(f,g)$-bicriteria approximation where~$f = \Omega(k^{1/r})$ and $g=(1+\delta)(2r+1)$. Our result shows that constructing the girth graph is essentially the ``correct'' way to analyse the greedy spanner, up to constant factors, since lower bounds on the girth of~$H$ directly translates to lower bounds on the dilation factor of the greedy algorithm for Problem~\ref{problem:bicriteria}. 

We divide our proof of Theorem~\ref{theorem:main_greedy_lower} into six parts. First, we define the underlying metric space~$M$. Second, we define the graph~$G$. Third, we define our Problem~\ref{problem:bicriteria} instance. Fourth, we upper bound~$t^*$ in our Problem~\ref{problem:bicriteria} instance. Fifth, we show that if~$t \leq gt^*$ in our Problem~\ref{problem:bicriteria} instance, then the greedy~\mbox{$t$-spanner} adds at least~$fk+1$ edges to~$G$. Sixth, we analyse the algorithm in Theorem~\ref{theorem:main_greedy_upper}.

\textbf{Part 1.} We define the underlying metric space~$M$. Assume that the vertices of the girth graph~$H$ are~$V(H) = \{w_1,\ldots,w_n\}$ and its edges are~$E(H) = \{e_1,\ldots,e_m\}$, where~$m = \Omega(n^{1+1/r})$. The vertices of~$M$ are~$V(M) = \{u_{i,j}:1 \leq i \leq n, 0 \leq j \leq m\}$. For an example of the vertices~$u_{i,j}$ where~$1 \leq i \leq 4$, and~$0 \leq j \leq 5$, see Figure~\ref{fig:4_greedy_lower_bound_01}. 

\begin{figure}[ht]
    \centering
    \includegraphics[width=\textwidth]{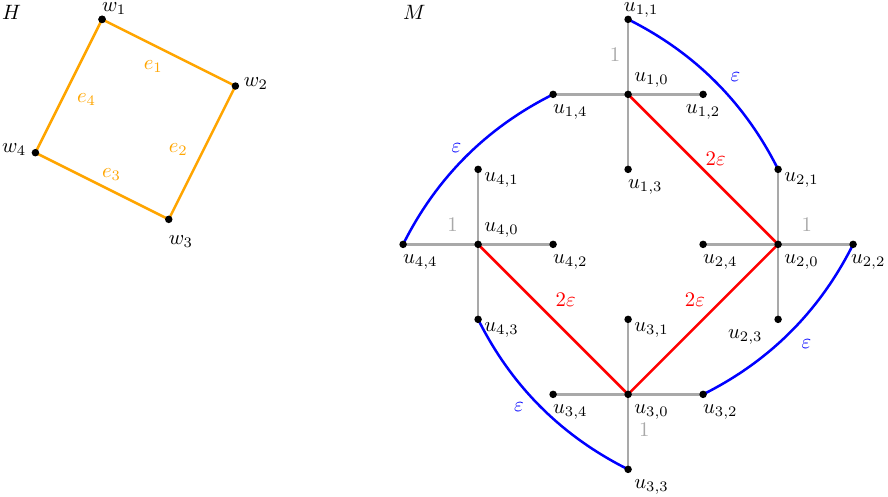}
    \caption{The girth graph~$H$ (orange), the vertices~$u_{i,j} \in V(M)$ (black), the edges~$M_1$ with length~1 (grey), the edges~$M_2$ with length~$2 \varepsilon$ (red), and the edges~$M_3$ with length $\varepsilon$ (blue).}
    \label{fig:4_greedy_lower_bound_01}
\end{figure}

The metric is a graph metric, where distances are shortest path distances in the graph~$M = (V(M),E(M))$. It remains to construct the edges~$E(M)$. We divide the edges~$E(M)$ into three subsets: $M_1$, $M_2$ and~$M_3$. Choose~$\varepsilon = 1/4rn$. Refer to Figure~\ref{fig:4_greedy_lower_bound_01}.

\begin{itemize}
    \item (Grey) Define $M_1 = \{u_{i,0} u_{i,j}: 1 \leq i \leq n,\ 1 \leq j \leq m \}$. Each edge in $M_1$ has length $1$.
    \item (Red) Define $M_2 = \{u_{i,0} u_{i+1,0}: 1 \leq i \leq n-1\}$. Each edge in $M_2$ has length $2 \varepsilon$.
    \item (Blue) Define $M_3 = \{u_{a,j} u_{b,j}: e_j = (w_a, w_b),\ 1 \leq j \leq m\}$. Each edge in $M_3$ has length $\varepsilon$. 
\end{itemize}
This completes the definition of the metric~$M$. 

\textbf{Part 2.} We define the graph~$G$. Let~$V(G) = V(M)$, and set~$E(G) = M_1$. In Figure~\ref{fig:4_greedy_lower_bound_01}, the graph~$G$ only uses the grey edges. 

\textbf{Part 3.} We define our Problem~\ref{problem:bicriteria} instance. Let~$M$ be the metric in part~1,~$G$ be the graph in part~2,~$k=n-1$ and~$f = (m-1)/(n-1)$. We will prove that the dilation parameter is at least~$g = 2r+1$. 

\textbf{Part 4.} We upper bound~$t^*$ in our Problem~\ref{problem:bicriteria} instance. Define~$S' = M_2$ and define~$G' = G \cup S'$. Note that $|S'| = n-1$. If the dilation of $G'$ is $t'$, then $t^* \leq t'$, since  $t^*$ is the minimum dilation if $k=n-1$ edges are added to $G$. It remains to upper bound $t'$. Note that~$d_{G'}(u_{i,0},u_{j,0}) = 2\varepsilon \cdot |i-j|$. If~$(i,a) \neq (j,b)$, then~$d_{G'}(u_{i,a},u_{j,b}) = 2\varepsilon \cdot |i-j| + 2$. Thus, the diameter of~$G'$ is at most~$2n\varepsilon + 2$. The metric distance between any pair of points in~$G'$ is at least~$\varepsilon$. 
So $t' = \max_{u,v \in G'} \frac{d_{G'}(u,v)} {d_M(u,v)} \leq \frac {2n\varepsilon + 2} \varepsilon$. Therefore, $t^* \leq t' \leq (2n\varepsilon + 2)/\varepsilon$, as required.

\textbf{Part 5.} We show that if~$t \leq gt^*$ in our Problem~\ref{problem:bicriteria} instance, then the greedy~\mbox{$t$-spanner} adds at least~$fk+1$ edges to~$G$. We prove by induction that the first~$fk+1$ edges added by the greedy~\mbox{$t$-spanner} are all edges from~$M_3$. Note that all edges in $M_3$ will be considered before any edge in $M_2$. The base case of~$i=0$ is trivially true. Assume the inductive hypothesis that the first~$i$ edges are from~$M_3$, where~$0 \leq i \leq fk$. Consider the graph~$G_i$, and an edge $e \in M_3$ not currently in~$G_i$. Then~$d_M(e) = \varepsilon$. Let~$e = u_{a,j} u_{b,j}$, where~$a \neq b$. Next, we compute~$d_{G_i}(e)$ by considering the shortest path between~$u_{a,j}$ and~$u_{b,j}$ in~$G_i$. If there is no shortest path, then~$d_{G_i}(e) = \infty$, and by Definition~\ref{definition:greedy_t_spanner}, the greedy~\mbox{$t$-spanner} would add either~$e$ or another edge with length~$\varepsilon$ to~$G_i$ to obtain~$G_{i+1}$. Otherwise, the shortest path must use edges in~$M_3$, since the edges in~$M_1$ alone cannot connect~$u_{a,j}$ to~$u_{b,j}$, since~$a \neq b$. Let the edges of~$M_3$ along the shortest path from~$u_{a,j}$ to~$u_{b,j}$, in order, be~$\{(u_{a_1,j_1},u_{b_1,j_1}), \ldots, (u_{a_d,j_d},u_{b_d,j_d})\}$ for some~$1 \leq d \leq m-1$. Note for all $1 \leq c \leq d$, each of the vertices~$u_{a_c,j_c},u_{b_c,j_c} \in V(G)$ are distinct, since no edges in~$M_3$ share an endpoint. Consider the path from~$u_{a,j}$ to~$u_{a_1,j_1}$. All edges along this path are in~$M_1$, and there are at least two edges, since~$j \neq 0$ and~$j_1 \neq 0$. Therefore,~$d_{G_i}(u_{a,j},u_{a_1,j_1}) \geq 2$, and~$a = a_1$. Applying the same argument between~$u_{b_d,j_d}$ and~$u_{b,j}$, we get~$d_{G_i}(u_{b_d,j_d}, u_{b,j}) \geq 2$. Applying the same argument between~$u_{b_c,j_c}$ and~$u_{a_{c+1},j_{c+1}}$ for all~$1 \leq c \leq d-1$, we get~$d_{G_i}(u_{b_c,j_c}, u_{a_{c+1},j_{c+1}}) \geq 2$. Summing this all together, the length of the shortest path from~$u_{a,j}$ to~$u_{b,j}$ is at least~$2(d+1) + 2d\varepsilon$. Next, we show~$d \geq 2r+1$. Since there is an edge~$u_{a,j}u_{b,j}$ in~$M_3$, there is an edge between~$w_a$ and~$w_b$ in~$H$. Since there is an edge~$u_{a_c,j_c}u_{b_c,j_c}$ in~$M_3$, there is an edge between~$w_{a_c}$ and~$w_{b_c}$ in~$H$. But there is a path from~$u_{a,j}$ to~$u_{a_1,j_1}$ using edges only in~$M_1$, so~$a = a_1$. Similarly,~$b_1 = a_2$,\ldots,~$b_{d-1} = a_d$, $b_d = b$. So there is a cycle~$w_{a_1},w_{a_2},\ldots,w_{a_d},w_{b_d}$ of length~$d+1$ in the graph~$H$. But the girth of~$H$ is~$2r+2$, so~$d \geq 2r+1$, as claimed. The length of the shortest path from~$u_{a,j}$ to~$u_{b,j}$ is at least~$2(d+1) + 2d\varepsilon$, which is at least~$2(2r+2) + 2(2r+1)\varepsilon$. But now, 
\[
    \begin{array}{rcl}
    d_{G_i}(u_{a,j},u_{b,j}) 
    &\geq& 2(2r+2) + 2(2r+1)\varepsilon 
    \\
    &\geq& (2r+1) \cdot \varepsilon \cdot (2/\varepsilon + \frac 1 {2r+1} \cdot (2/\varepsilon) + 2)
    \\
    &>& g \cdot d_M(u_{a,j},u_{b,j}) \cdot (2/\varepsilon + 2n)
    \\
    &=& gt^* \cdot d_M(u_{a,j},u_{b,j})
    \\ 
    &\geq& t \cdot d_M(u_{a,j},u_{b,j}).
    \end{array}
\]
Therefore, $d_{G_i}(u_{a,j},u_{b,j}) > t \cdot d_M(u_{a,j},u_{b,j})$, so by Definition~\ref{definition:greedy_t_spanner}, the greedy~\mbox{$t$-spanner} would add either~$u_{a,j}u_{b,j}$ or another edge in~$M_3$ to~$G_i$ to obtain~$G_{i+1}$. This completes the induction, so the first~$fk+1$ edges added by the greedy~$t$-spanner are all edges from~$M_3$. As a consequence,~$t \leq g t^*$ implies that the greedy~\mbox{$t$-spanner} adds at least~$fk+1$ edges to~$G$, completing the proof.

\textbf{Part 6.} We analyse the algorithm in Theorem~\ref{theorem:main_greedy_upper}. If at most~$fk$ edges are added, then we continue searching for dilation values less than~$t$, whereas if at least~$fk+1$ edges are added, then we continue searching for dilation values greater than~$t$. Therefore, for all~$t \leq gt^*$, the algorithm in Theorem~\ref{theorem:main_greedy_upper} would continue searching for values greater than~$t$. The dilation value returned by the algorithm is at least~$gt^*$. So Theorem~\ref{theorem:main_greedy_upper} returns an~$(f,(1+\delta)g)$-bicriteria approximation, where~$f = m-1 / n-1 = \Omega(k^{1/r})$, and~$g = 2r+1$, as required.

Finally, putting all six parts together, we obtain Theorem~\ref{theorem:main_greedy_lower}. We also provide a corollary to Theorem~\ref{theorem:main_greedy_lower} which does not assume the Erd\"os girth conjecture.

\begin{corollary}
For all $r \geq 1$, there is a graph class for which the algorithm in Theorem~\ref{theorem:main_greedy_upper} returns an $(f,g)$-bicriteria approximation, where
$$f = \Omega(H(k,r)/k)\quad and \quad g = 2r+1,$$
and $H(k,r)$ is the maximum number of edges in an $k$-node graph with girth $2r+2$.
\end{corollary}

\section{Set cover reduction}
\label{sec:set_cover_lower}

In this section, we will prove Theorem~\ref{theorem:main_set_cover_lower}. We restate the theorem for convenience.

\mainsetcoverlower*

We reduce from set cover. The set cover instance consists of~$m$ elements~$E = \{e_1,\ldots,e_m\}$ and~$L$ sets~$S = \{S_1,\ldots,S_L\}$, where~$S_\ell \subseteq E$ for all~$1 \leq \ell \leq L$. Let~$n = O(mL)$ denote the complexity of the set cover input.

\begin{fact}
    \label{fact:fpt_hardness}
    Let $(S,E)$ be a set cover instance with complexity $n$, and let $k \in \mathbb N$. Assuming W[1] $\neq$ FPT, for any computable function $h$, there is no $\ h(k) \cdot \poly(n)\ $ time algorithm to distinguish whether
    \begin{itemize}
        \item There is a set cover of $(S,E)$ that has size at most $k$, or
        \item Every set cover of $(S,E)$ has size at least $h(k) \cdot k$.
    \end{itemize}
\end{fact}

\begin{proof}
    The above fact follows from Theorem~1.3 in~\cite{DBLP:journals/jacm/SLM19}, Footnote~8 in~\cite{DBLP:journals/jacm/SLM19}, and the gap-preserving reduction from $k$-dominating-set to $k$-set-cover~\cite{downey1995parameterized}.
\end{proof}

We summarise our approach. We show that every set cover instance can be reduced to a Problem~\ref{problem:bicriteria} instance. If there is an~$(f,g)$-bicriteria approximation for Problem~\ref{problem:bicriteria} that can be computed in polynomial time, where~$f = h(k)$ and~$g = 2-\varepsilon$, then one would be able to approximate set cover to within a factor of~$h(k)$ in polynomial time, contradicting Fact~\ref{fact:fpt_hardness}.

We divide our proof of Theorem~\ref{theorem:main_set_cover_lower} into six parts. First, we define the underlying metric space~$M$. Second, we define the graph~$G$. Third, we define our Problem~\ref{problem:bicriteria} instance. Fourth, we upper bound the dilation, assuming the set cover instance is a \mbox{YES-instance}. Fifth, we lower bound the dilation, assuming the set cover instance is a \mbox{NO-instance}. Sixth, we show that it is \mbox{W[1]-hard} to obtain an~$(f,g)$-bicriteria approximation. 

\textbf{Part 1.} We define the underlying metric space~$M$. For each element~$e_i$ where~$1 \leq i \leq m$, construct~$2k+2$ vertices~$U_i = \{u_{i,1},u_{i,1}',\ldots,u_{i,k+1},u_{i,k+1}'\}$. For each set~$S_\ell$ where~$1 \leq \ell \leq L$, construct three vertices~$V_\ell = \{v_\ell, v_\ell', w_\ell\}$. For an example of the vertices~$u_{i,j}, u_{i,j}', v_\ell, v_\ell', w_\ell$ where~$1 \leq i \leq 3$, $1 \leq j \leq 2$, and~$1 \leq \ell \leq 2$, see Figure~\ref{fig:5_set_cover_lower_01}. Define the vertices~$V(M) = U_1 \cup \ldots U_m \cup V_1 \cup \ldots V_L$. 

\begin{figure}[ht]
    \centering
    \includegraphics{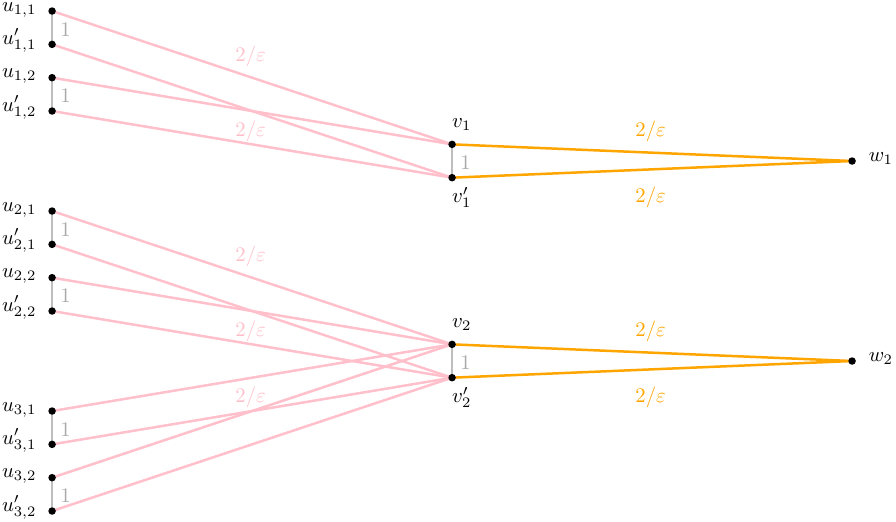}
    \caption{The graph metric~$M$ for a set cover instance with elements $\{e_1,e_2,e_3\}$ and sets~$S_1 = \{e_1\}$ and~$S_2 = \{e_2,e_3\}$. The edges~$M_1 \cup M_2$ are shown in grey, and have length~1. The edges~$M_3 \cup M_4$ are shown in pink, and have length~$2/\varepsilon$. The edges~$M_5 \cup M_6$ are shown in orange, and have length~$2/\varepsilon$.}
    \label{fig:5_set_cover_lower_01}
\end{figure}

The metric distances between any pair of vertices in~$V(M)$ is the shortest path between them in the weighted graph~$M = (V(M),E(M))$. We divide the edges~$E(M)$ into six subsets:~$M_1,\ldots,M_6$. Refer to Figure~\ref{fig:5_set_cover_lower_01}. 

\begin{itemize}
    \item (Grey) Define~$M_1 = \{u_{i,j}u_{i,j}': 1 \leq i \leq m,\ 1 \leq j \leq k+1\}$. Each edge in~$M_1$ has length~1.
    \item (Grey) Define~$M_2 = \{v_{\ell}v_{\ell}': 1 \leq \ell \leq L\}$. Each edge in~$M_2$ has length~1.
    \item (Pink) Define~$M_3 = \{u_{i,j}v_\ell: e_i \in S_\ell,\ 1 \leq j \leq k+1\}$. Each edge in~$M_3$ has length~$2/\varepsilon$.
    \item (Pink) Define~$M_4 = \{u_{i,j}'v_\ell': e_i \in S_\ell,\ 1 \leq j \leq k+1\}$. Each edge in~$M_4$ has length~$2/\varepsilon$.
    \item (Orange) Define~$M_5 = \{v_\ell w_\ell: 1 \leq \ell \leq L\}$. Each edge in~$M_5$ has length~$2/\varepsilon$.
    \item (Orange) Define~$M_6 = \{v_\ell' w_\ell: 1 \leq \ell \leq L\}$. Each edge in~$M_6$ has length~$2/\varepsilon$.
\end{itemize}
This completes the definition of the metric~$M$. 

\textbf{Part 2.} We define the graph~$G$. Define the vertices of~$G$ to be~$V(G) = V(M)$. Define the edges of~$G$ to be~$E(G) = M_3 \cup M_4 \cup M_5 \cup M_6$. This completes the definition of~$G$. 
In Figure~\ref{fig:5_set_cover_lower_01}, the graph~$G$ uses only the pink and orange edges.

\textbf{Part 3.} We define our Problem~\ref{problem:bicriteria} instance. The metric space~$M$, graph~$G$, and parameter~$k$ are defined as above. We will show that if the set cover instance is a \mbox{YES-instance}, then there are~$k$ edges that one can add to the graph so the resulting dilation is~$t^* \leq 4/\varepsilon + 1$. We will show that if the set cover instance is a \mbox{NO-instance}, then there are no~$fk$ edges that one can add to the graph so that the resulting dilation is at most~$gt^*$, where~$f = (1-\alpha) \log n$ and~$g = 2-\varepsilon$. This completes the definition of the Problem~\ref{problem:bicriteria} instance.

\textbf{Part 4.} We show~$t^* \leq 4/\varepsilon + 1$, assuming the set cover instance is a \mbox{YES-instance}. In particular, there exists~$k$ sets~$\{S_{\ell_1},\ldots,S_{\ell_k}\}$ in~$S$ that covers~$E$. Define~$S' = \{v_{\ell_1}v_{\ell_1}',\ldots,v_{\ell_k}v_{\ell_k}'\} \subseteq M_2$ contain the $k$ edges in $M_2$ associated with the $k$ sets in the YES-instance. Define~$G' = G \cup S'$. Let $t'$ be the dilation of $G'$. Since $|S'| = k$, we have $t^* \leq t'$. For all graph metrics~$M$, the maximum dilation is obtained between a pair of points where the shortest path in~$M$ between them only uses a single edge in~$E(M)$. This is because, if the shortest path has multiple edges in~$E(M)$, then the dilation between the endpoints of one of those edges would be at least as large. Therefore, it suffices to consider the dilation between the endpoints of the edges in~$M_1,\ldots,M_6$. However, the dilation between the endpoints of edges in~$M_3,\ldots,M_6$ is~$1$. Therefore, it suffices to consider endpoints of edges in~$M_1$ and~$M_2$. For an edge $v_\ell v_\ell' \in M_2$, if $v_\ell v_\ell' \in S'$ then its dilation is 1. However, if $v_\ell v_\ell'$ is not in $S'$, then we can upper bound the distance with the triangle inequality via $w_\ell$, i.e.
\[
    d_{G^*}(v_\ell,v_\ell') \leq d_{G^*}(v_\ell,w_\ell) + d_{G^*}(w_\ell,v_\ell') \leq 4/\varepsilon.
\]
For~$M_1$, recall that~$\{S_{\ell_1},\ldots,S_{\ell_k}\}$ covers~$E$. So, for each~$u_{i,j}u_{i,j}' \in M_1$ there exists an~$\ell$ satisfying~$u_{i,j}v_\ell \in M_3$, $v_\ell v_\ell' \in S^*$, and~$v_\ell' u_{i,j}' \in M_4$. Therefore,
\[
    d_{G^*}(u_{i,j},u_{i,j}') \leq d_{G^*}(u_{i,j},v_\ell) + d_{G^*}(v_\ell,v_\ell') + d_{G^*}(v_\ell',u_{i,j}') \leq 4/\varepsilon + 1.
\]
Hence, the dilation of~$G^*$ is at most~$4/\varepsilon + 1$, completing the proof that~$t^* \leq 4/\varepsilon + 1$.

\textbf{Part 5.} We show that there are no~$fk$ edges that one can add to~$G$ to obtain a dilation of at most~$gt^*$, assuming the set cover instance is a \mbox{NO-instance}. Recall that~$f = (1-\alpha) \log n$ and~$g = 2 - \varepsilon$. Suppose for the sake of contradiction that there exists~$fk$ edges so that the final dilation is at most~$gt^* = (2-\varepsilon)(4/\varepsilon + 1) < 8/\varepsilon$. Let the set of~$fk$ edges be~$F$. Suppose that one of the edges in~$F$ is from the set~$M_1$. Let the edge be~$u_{i,j}u_{i,j}' \in F \cap M_1$. For all~$e_i \in E$ there exists an~$S_\ell$ so that~$e_i \in S_\ell$. We exchange the edge~$u_{i,j}u_{i,j}'$ with~$v_\ell v_\ell'$. We show that this exchange does not affect the property that the final dilation is at most~$gt^* < 8/\varepsilon$. Similarly to in part~4, it suffices to consider pairs of points that are endpoints of~$M_1 \cup M_2$. For~$M_2$, any path between~$v_p$ and~$v_p'$ that uses the edge~$u_{i,j} u_{i,j}'$ must pass through~$v_\ell$ and~$v_{\ell}'$, so exchanging the edge~$u_{i,j}u_{i,j}'$ with~$v_\ell v_\ell'$ decreases the shortest path distance between all~$v_p v_p' \in M_2$. For $M_1$, any path from~$u_{p,q}$ to~$u_{p,q}'$ that uses the edge~$u_{i,j}u_{i,j}'$ has length at least~$8/\varepsilon+1$, so it cannot be the shortest path between~$u_{p,q}u_{p,q}' \in M_1$. Therefore, after performing exchanges for all~$u_{i,j}u_{i,j}' \in F \cap M_1$, we can assume that~$F \subseteq M_2$. Since~$|F| = k (1-\alpha) \log n$ and the set cover instance is a \mbox{NO-instance}, the set~$F$ cannot be a set cover for the elements~$e_i \in E$. Therefore, there exists an~$i$ so that~$v_\ell v_\ell' \not \in F$ for all~$S_\ell$ where~$e_i \in S_\ell$. For all~$v_\ell v_\ell' \not \in F$, we have~$d_{G \cup F}(v_\ell, v_\ell') \geq d_{G \cup F}(v_\ell, w_\ell) + d_{G \cup F}(w_\ell, v_\ell') = 4/\varepsilon$. Therefore,
\[
    d_{G \cup F}(u_{i,j},u_{i,j}') \geq \max_{\ell: e_i \in S_\ell} \left( d_{G \cup F}(u_{i,j},v_\ell)+d_{G \cup F}(v_\ell,v_\ell')+d_{G \cup F}(v_\ell',u_{i,j}') \right) \geq 8 / \varepsilon,
\]
contradicting the fact that the final dilation is at most~$gt^*$. This completes the proof that there are no~$fk$ edges that one can add to~$G$ to obtain a dilation of at most~$gt^*$, assuming the set cover instance is a \mbox{NO-instance}.

\textbf{Part 6.} We show that it is \mbox{W[1]-hard} to obtain an~$(f,g)$-bicriteria approximation. If there is an~$(f,g)$-bicriteria approximation for Problem~\ref{problem:bicriteria}, one would be able to decide whether~$(i)$ there exists a set of~$k$ edges to add to~$G$ to obtain a dilation of~$t^*$, or~$(ii)$ there are no~$fk$ edges that one can add to~$G$ to obtain a dilation of at most~$gt^*$. Therefore, one would be able to decide whether the set cover instance is a \mbox{YES-instance} or a \mbox{NO-instance}. It follows from Fact~\ref{fact:fpt_hardness} that, assuming FPT $\neq$ W[1], one cannot obtain an~$(f,g)$-bicriteria approximation for~$f = h(k)$ and~$g = 2-\varepsilon$, where $h$ is any computable function and $\varepsilon > 0$.

Finally, putting all six parts together, we obtain Theorem~\ref{theorem:main_set_cover_lower}, as required.

\section{Conclusion}

We provide bicriteria approximation algorithms for the problem of adding~$k$ edges to a graph to minimise its dilation. Our main result is a~$(2 \sqrt[r]{2} \ k^{1/r},2r)$-bicriteria approximation for all~$r \geq 1$, that runs in $O(n^3 \log n)$ time. Our analysis is tight and it is \mbox{W[1]-hard} to obtain a~$(h(k),2-\varepsilon)$-bicriteria approximation for any computable function~$h$ and for any~$\varepsilon > 0$. We provide a simple~$(4k^2 \log n,1)$-bicriteria approximation.

We conclude with directions for future work. Problem~\ref{problem:tree_spanner} remains open. In particular, Obstacle~\ref{obstacle:tree_spanner} asks: is there an~$\varepsilon > 0$ for which there is an~$O(n^{1-\varepsilon})$-approximation algorithm for the minimum dilation spanning tree problem? The linear approximation factor of Problem~\ref{problem:graph_augmentation} cannot be improved unless Obstacle~\ref{obstacle:graph_augmentation} is resolved. An alternative way to circumvent Obstacle~\ref{obstacle:graph_augmentation} is to consider Problem~\ref{problem:graph_augmentation} in the special case of unweighted graph metrics. Finally, Problem~\ref{problem:bicriteria} offers several directions for future work. Can one obtain a trade-off between sparsity and dilation that is better than the greedy \mbox{$t$-spanner} construction? What is the sparsity-dilation trade-off when~$1 < f < 2$? Can the approximation factor or the running time of Theorem~\ref{theorem:main_set_cover_upper} be improved?

\bibliographystyle{plain}
\bibliography{bib}

\end{document}